\documentclass{amsart}

\usepackage[dvipdfm]{graphicx}

\usepackage{amsmath}
\usepackage{graphicx}
\usepackage{amsfonts}
\usepackage{amssymb}
\usepackage{hyperref}

\newtheorem{corollary}{Corollary}[section]
\newtheorem{theorem}{Theorem}[section]
\newtheorem{lemma}{Lemma}[section]
\newtheorem{proposition}{Proposition}[section]
\newtheorem{remark}{Remark}[section]
\newtheorem{definition}{Definition}[section]
\newtheorem{example}{Example}[section]

\newcommand{\Ex}{\mathbf{E}}

\title[Defaultable Bonds via HKA]
      {Defaultable Bonds via HKA}

\author[Y\^uta Inoue and Takahiro Tsuchiya]{}

\begin{document}

\maketitle

\centerline{\scshape  Y\^uta Inoue$^1$ and Takahiro Tsuchiya$^2$}
 \medskip
{\footnotesize
\centerline{$^1$Graduate School of Mathematics}
\centerline{Ritsumeikan University} 
\centerline{1-1-1 Nojihigashi, Kusatsu, Shiga 525-8577, Japan}
\centerline{shinzo.yuta@gmail.com} }
\medskip

{\footnotesize
\centerline{$^2$Department of Mathematical Sciences}
\centerline{Ritsumeikan University}
\centerline{1-1-1 Nojihigashi, Kusatsu, Shiga 525-8577, Japan}
\centerline{suci@fc.ritsumei.ac.jp} }
\medskip



\medskip

\begin{abstract}

{\em To construct a no-arbitrage defaultable bond market, 
we work on the state price density framework. 
Using the heat kernel approach (HKA for short) 
with the killing of a Markov process, 
we construct a {\it single} defaultable bond market that 
enables an explicit expression of a defaultable bond and credit spread 
under quadratic Gaussian settings.
Some simulation results show that 
the model is not only tractable but realistic. }\\[0.5cm]
\end{abstract}

\section{Introduction}
The HKA, which is an abbreviation of ``Heat Kernel Approach
to interest rate modelling, was introduced by one of the authors 
and his collaborators in \cite{HKA}. 
Briefly speaking, HKA is a systematic method to produce
a tractable interest rate model which is ``Markov functional"
in the sense of Hunt-Kennedy-Pelsser \cite{HKP}. 
In the fundamental paper \cite{HKA}, 
four different types of implementation methods are 
introduced. Namely, 
1) Eigenfunction models
2) Weighted HKA
3) Killed HKA
and 4) Trace Approach. 
As is pointed out in \cite{HKA}, 
the eigenfunction models are tailor-made
for swaption pricing,
and a deeper understanding for its mathematical
structure leads to the trace approach, 
which is mathematically most involved.  
The weighted HKA is extended to 
time-inhomogeneous setting and applied to information-based
models by J. Akahori and A. Macrina \cite{AM:10}. 

In the present paper, we will demonstrate
how the Killed HKA is applied to 
the modelling of defaultable bonds
by constructing a market where the market price of risk 
and the default probability are ``built in the same block"
(whose precise meaning will be given later). 
We stress that the HKA is basically a state-price density approach
where everything is written under the physical= statistical 
measure. Since the HKA furthermore 
gives an analytically tractable model in nature, 
the framework proposed in this paper would be promising 
in respect of modelling defaultable markets. 

The organization of the present paper is as follows. 
After recalling the {\em plain-vanilla} HKA in
section \ref{PVHKA}
and the killed HKA in section \ref{KHKA}, 
we shall give the main result, a framework
with in the Killed HKA to model a defaultable bond market
in section \ref{FMW}. 
In section \ref{QE}, we will give some simulation
results of an explicit example with a quadratic form of Wiener process.

\section*{Acknowledgement}
The authors are 
deeply grateful to Professor Dr. Jiro Akahori. 
His insightful comments and suggestions 
were an enormous help to us. 

\section{Heat Kernel Approach}
Here we briefly recall the approach.
\subsection{Plain-Vanilla HKA}\label{PVHKA}
We work on a probability space 
$ (\Omega,\mathcal{F},{\mathsf P}) $ 
with filtration $ { \{ \mathcal{F}_t \} }_{t \geq 0} $. 
Now we consider a general Markov process 
$ { \{ X^x_t \} }_{t \geq 0,x \in \mathcal{S}} $ 
on a polish space $ \mathcal{S} $.

\begin{definition}
Let $X$ be an $\mathcal{S}$-valued Markov process. 
We shall say that 
a function $p$ satisfies the {\em propagation property\/} if  
\begin{equation}\label{propro}
\Ex [p(t,X_s^x)]=p(t+s,x)
\end{equation}
holds for any $t,s \geq 0$ and $x \in \mathcal{S}$.
\end{definition}
The following fact is initialized by \cite{AT:06} and 
developed in \cite{HKA}.
\begin{proposition}[Akahori et al. \cite{HKA}]
Let $X$ be a $\mathcal{S}$-valued Markov process,
$ \lambda $ be a positive function on the half line,  
and $p$ be a function with the propagation property. 
The bonds market given by 
\begin{equation}\label{bond1111}
P_f(t,T) 
= \frac{p(\lambda_T + T-t,X^x_t)}{p(\lambda_t,X^x_t)}
\end{equation}
is an arbitrage free market.
\end{proposition}

\begin{example}[Generic Example]
Take a measurable, bounded $ h: \mathcal{S} \to \mathbb{R}_{\geq 0} $, then
\begin{equation}\label{generic}
p(t,x):= \Ex [ h(X^x_t) ]
\end{equation}
satisfies the propagation property (\ref{propro}).
In fact, by the Markov property, we have
\begin{equation*}
\begin{split}
\Ex \left[p(t,X_s) \right]&=\Ex \left[ \Ex \left[h(X^{X^x_s}_t) \right] 
\right]
= \Ex \left[\Ex \left[ h(X^x_{t+s}) | \mathcal{F}^X_s \right] \right]\\
&=\Ex \left[h(X^x_{t+s}) \right]
=p(t+s,x).
\end{split}
\end{equation*}
\end{example}

\begin{proof}
Let $\pi_t = p( \lambda_t, X^x_t )$.
By the propagation property of $p$ and the Markov property
of $ X $, we have
\begin{equation*}
\begin{split}
P_f(t,T)&=\frac{\Ex \left[ \pi_T | \mathcal{F}_t \right]}{\pi_t}
=\frac{\Ex \left[ p(\lambda_T,X^x_T) | \mathcal{F}_t \right]}
{p(\lambda_t,X^x_t)}\\
&=\frac{\Ex \left[ p(\lambda_T ,X^{X^x_t}_{T-t}) \right]}
{p(\lambda_t,X^x_t)}
=\frac{p(\lambda_T+ T-t,X^x_t)}{p(\lambda_t,X^x_t)}.
\end{split}
\end{equation*}
This means $\pi$ is the state price density of the market.
\end{proof}

It should be noted that we do not assume $ \pi_t $
to be a supermartingale in [\ref{bond1111}], i.e. in economic terms 
\emph{we do not assume positive short rates}.
The four implementation methods mentioned in the introduction is 
introduced in \cite{HKA} to obtain supermartingales out of a propagator, 
or equivalently to obtain positive rate models.

\subsection{The Killed HKA}\label{KHKA}
We then recall, and give a more detailed description to,
the Killed HKA\footnote{This part is shared with \cite{TI}.}.
Let $V$ be a non-negative measurable function on $\mathcal{S}$. 
Put $ Y_t^y = y + \int_0^t V (X^x_s)\,ds $ for $ y \in \mathbf{R} $. 
Let us define 
\begin{equation}\label{qtx}
q (t, x ) =  \Ex [ \exp{ ( - \int_0^t V (X^x_s) \,ds ) } ].
\end{equation}
Then the function 
\begin{equation*}
q (t, x, y ) = e^{-y} q (t, x ),  
\end{equation*}
satisfies propagation property with respect to $(X^x, Y^y )$; 
\begin{equation}\label{pp}
\Ex [q (s, X^x_t, Y^y_t )]= q (t+s, x,y). 
\end{equation}
In fact, by the Markov property of $X$, 
we have 
\begin{equation*}
\begin{split}
\Ex [e^{-Y_{t+s}} | \mathcal{F}_t]
&=\Ex [e^{- Y_{s} ( \theta_t )} | \mathcal{F}_t]
\times e^{- Y_{t} } 
=\Ex [e^{-Y_{s}} | X_t]
\times e^{- Y_{t} },
\end{split}
\end{equation*}
where $\theta$ is the shift operator. 
Thus we obtain 
\begin{equation*}
\begin{split}
\Ex [q (s, X^x_t, Y^y_t )]
=\Ex [e^{-Y^y_t} q (s, X^x_t)] 
=\Ex [e^{-Y^y_t} e^{  - Y_{s}^y \circ \theta_t  }] 
= q (t+s, x,y). 
\end{split}
\end{equation*}
This fact ensures that the bond market model
constructed as
\begin{equation}
P(t,T) = \frac{q(\lambda_T+T-t, X^x_t)}{q(\lambda_t,X^x_t)},
\end{equation}
where $ \lambda $ is an increasing function, 
is arbitrage-free since we can choose
\begin{equation*}
\pi_t = q( \lambda_t, X_t ) \exp (-\int_0^{t} V(X_s)\,ds)  
\end{equation*}
as a state price density of the market. 
In fact, 
\begin{equation*}
\begin{split}
& \Ex [\pi_T |\mathcal{F}_t] =
\Ex [ \,\, \Ex [ \exp{ ( - \int_{T}^{T+\lambda_T} 
V (X^x_u) \,du ) } |\mathcal{F}_T] \exp (-\int_0^{T} V(X_s)\,ds)\,\,|\mathcal{F}_t] \\
&= \Ex [\exp{ ( - \int_{t}^{T+\lambda_T} V (X^x_u) \,du ) } \,\,|X_t] 
\exp (-\int_0^{t} V(X_s)\,ds) \\
& = q(\lambda_T+T-t, X_t) \exp (-\int_0^t V(X_s)\,ds) 
= \pi_t \frac{q(\lambda_T+T-t, X_t)}{q(\lambda_t,X_t)}. 
\end{split}
\end{equation*}

Note that the bond price $ P $ is decreasing in $ T $
since $ q $ is increasing in $ t $, 
which is ensured by the positivity of $ V $.
Thus we obtain a positive rate model.

\section{HKA to Defaultable Bond}\label{FMW}
This section is the main part of the present paper. 
Let us now consider a defaultable bond in the following situation:
\begin{enumerate}
\item The bond pays a unit account at the maturity $ T $ unless it defaults.
\item At the default time $ \tau $, nothing will be recovered. 
\item The state variable is a Markov process $ \{ X^x_t; t \geq 0 \} $, 
which is observable in the market.
\item The default probability is completely determined through 
the information of $ X $ in the following manner;
the hazard rate of the default time on the filtration $\mathcal{F}^X$
is given by 
$ \displaystyle
\Ex [1_{\{\tau > t \} } | \mathcal{F}^X_t] 
= \exp \left(- \int_0^t V (X^x_u) \,du \right) $ 
where 
$ \mathcal{F}^X $ is the natural filtration on $ X $ and 
$ V $ is a non-negative measurable function. 
\item The default come as a ``surprise" to the market.
To be precise, the market filtration $  \{\mathcal{G}_t \} $ 
is defined as 
$ \mathcal{G}_t 
= \sigma (X_s ,  \{ \tau \leq s \} ;  s\leq t )$ 
and assume that $ \mathcal{F}^X_0 = \{ \Omega, \emptyset \} = \mathcal{G}_0 $. 

\item A state price density of the market is given by 
$\pi_t:=q( \lambda_t , X_t ) 
\displaystyle
=\Ex \left[ \exp \left(- \int_0^{\lambda_t} V (X^x_u) \,du \right) \right]$ 
where $q$ is defined as $(\ref{qtx})$ and 
$ \lambda $ is a non-decreasing function.
\end{enumerate}

Note that the assumptions 
1--5 may be natural (except assumption 2, which 
assumes zero recovery) and very generic, 
while the last assumption is very specific in that 
the function $ V $ controls both the market price of a risk
as well as the default probability of a bond. 
Very heuristically speaking, this market is fully subject to 
the risk of a defaultable bond. 

We stress that this is just a toy model, which exhibits 
how the killed HKA is applied to a defaultable market modeling. 
The following is established in \cite{TI}: 
\begin{theorem}\label{Main}
Under the above assumptions 1--6, 
\begin{itemize}
\item[(i)] the price $P_d(t,T)$ of a defaultable zero coupon bond is given by 
\begin{equation}\label{defaultable}
P_d(t,T) 
= 
1_{ \{\tau > t\}} 
\frac{q ( \lambda_T + T-t,  X^x_t )}{q ( \lambda_t, X^x_t)},
\end{equation}
\item[(ii)] the price $P_f(t,T)$ of a default-free bond 
is given by 
\begin{equation}\label{dble}
P_f(t,T) =  
\frac{\hat{q}({\lambda_T + T-t},{T-t},X^x_t) }{q (\lambda_t, X^x_t)},
\end{equation}
where 
\begin{equation}\label{q^tx}
\hat{q}(t, s, x)=
 \Ex \left[ \exp \left( -\int_{s}^{t} V (X^x_u) \,du \right) \right], 
\end{equation}
\item[(iii)] and then the ``credit spread" is given by
\begin{equation}\label{crspread}
\partial_T \log 
\frac{\hat{q}({\lambda_T + T-t},{T-t},X^x_t)}{\hat{q}(\lambda_T + T-t,0, X^x_t) }. 
\end{equation}
\end{itemize}
\end{theorem}
\begin{remark}
Note that the ``credit spread" makes no sense when $\tau \le t$, 
so we can only think of the case that $\tau > t$.
\end{remark}
\begin{proof}
The proof is based on the following fundamental lemma
due to Dellacherie (see \cite{D}):
For any $\mathcal{F}_T^X$-integrable random variable 
$ Z $ and $ 0<t <T$, we have
\begin{equation*}
\Ex [1_{ \{ \tau > T \} } Z |\mathcal{G}_t] 
= 
\frac{1_{\{ \tau> t\} } }{ E[1_{\{ \tau> t\} }| \mathcal{F}^X_t] }
\Ex [ 1_{ \{ \tau > T \} } Z | \mathcal{F}^X_t].
\end{equation*}
Hence, we have
\begin{equation*}
P_d(t,T)
=\frac{ 1 }{ \pi_t }\,\, \frac{ 1_{ \{\tau > t\} } }
{ \Ex \left[ 1_{ \{\tau > t\} } | \mathcal{F}^X_t \right] }\
\,\,\Ex \left[ 1_{ \{\tau > T\} } \pi_T | \mathcal{F}^X_t \right].
\end{equation*}
Then by a Markov property and a Tower property,
\begin{equation*}
\Ex \left[ 1_{ \{\tau > T\} } \pi_T | \mathcal{F}^X_t \right]
= \Ex \left[\,\, \Ex \left[ 1_{ \{\tau > T\} } | \mathcal{F}^X_T \right]
q(\lambda_T,X^x_T)\, |\, \mathcal{F}^X_t \right],
\end{equation*}
and by the assumption that 
$ \displaystyle \Ex [\,1_{\{\tau > t \} }\, |\, \mathcal{F}^X_t] = 
\exp \left( -\int_{0}^{t} V (X^x_u) \,du \right)$
\begin{equation*}
\Ex \left[\, \Ex \left[\, 1_{ \{\tau > T\} }\, |\, \mathcal{F}^X_T \right]
q(\lambda_T,X^x_T)\, |\, \mathcal{F}^X_t \right] 
= \exp \left( -\int_{0}^{t} V (X^x_u) \,du \right) 
\Ex \left[\, e^{ -\int_t^T V(X^x_s) \,ds } 
q(\lambda_T,X^x_T)\, |\, \mathcal{F}^X_t \right].
\end{equation*}
Here applying the fact of the equation $(\ref{pp})$, we have
\begin{equation*}
\Ex \left[ e^{ -\int_t^T V(X^x_s) \,ds } 
q(\lambda_T,X^x_T) | \mathcal{F}^X_t \right]
=q(\lambda_T + T -t, X^x_t),
\end{equation*}
so that
\begin{equation*}
P_d(t,T)
= 1_{ \{\tau > t\} }
\frac{q(\lambda_T + T - t, X^x_t)}{ q(\lambda_t, X^x_t) }.
\end{equation*} 
On the other hand, (ii) follows a Markov property and a Tower property. And  
it is known that the ``credit spread" is given by
\begin{equation*}
-\partial_T \log \frac{ P_{d}(t,T) }{ P_{f}(t,T)}.
\end{equation*}
Here since $P_d(t,T)$, $P_f(t,T)$ are (i), (ii) respectively, 
\begin{equation*}
-\partial_T \log \frac{ P_{d}(t,T) }{ P_{f}(t,T)}
= \partial_T \log \frac{ \Ex \left[ q(\lambda_T, X^x_T) | 
\mathcal{F}^X_t \right] }{ q(\lambda_T + T - t, X^x_t) }\
\end{equation*}
when $\tau>t$. Then by a Markov property and a Tower property,
\begin{equation*}
\Ex \left[ q(\lambda_T, X^x_T) | \mathcal{F}^X_t \right]
=\Ex \left[ e^{ -\int_{T-t}^{\lambda_T + T - t} V(X^{X^x_t}_s) \,ds } \right]
= \hat{q}(\lambda_T + T - t, T-t, X^x_t).
\end{equation*}
Hence, we obtain
\begin{equation*}
-\partial_T \log \frac{ P_{d}(t,T) }{ P_{f}(t,T)}
= \partial_T \log \frac{\hat{q}({\lambda_T + T-t},{T-t},X^x_t)}
{q(\lambda_T + T-t, X^x_t) }.
\end{equation*}
\end{proof}

\section{Quadratic Example}\label{QE}
Now we give some simulation results of 
an explicit example, 
where $ X $ is a d-dimensional Wiener process 
and $ V (x) = \frac{ \beta^2 | x |^2}{2} $ ($ \beta >0 $).
Let $q(t,x)$ and $\hat{q}(t,x)$ be as in (\ref{qtx}), (\ref{q^tx}), 
then they are explicitly given by
\begin{equation}\label{qtxofqc}
q(t,x) = \left( \cosh \beta t  \right)^{-d/2}
\exp \left(-\frac{\beta x^2}{2} 
\frac{ \sinh \beta t }{  \cosh \beta t  } \right),
\end{equation}
and
\begin{equation}\label{q^txofqc}
\hat{q}(t,x)
=\left(\, \cosh \beta(t-s) + \beta s \sinh{ \beta(t-s) }\,  \right)^{-d/2}
\exp \left(-\frac{ \beta x^2 }{ 2 } \frac{ \tanh{ \beta(t-s) } }
{ 1 + \beta s \tanh{ \beta(t-s) } } \right),
\end{equation}
which result from Lemma \ref{lemma:d-quadratic}, Corollary \ref{Corollary:d-quadratic} in the Appendix. Hence, we obtain the analytic expression of the bond prices.
The following simulated yield curves (Fig.\ref{pfcurve@x}) and (Fig.\ref{pfcurve@beta}) implied by a default-free bond as 
\begin{equation*}
-\frac{ 1 }{ T } \log{ P_f(t, T) }
\end{equation*}
are obtained by using (\ref{dble}) and (\ref{q^txofqc}). 
Here the parameters are set to be $\beta = 0.1$, $x = 0.01, 10, 20, 30$, $\lambda_t = e^{t}/10$, and the present time $t = 0$ in (Fig.\ref{pfcurve@x}), 
$\beta = 1.8$ and $\lambda_t = e^{t}/100$ in (Fig.\ref{pfcurve@beta}).
Then $x$-axis stands for the maturities ranging from one year to ten years 
and $y$-axis does for the price of a default-free bond. 

\begin{figure}[htbp]
\begin{tabular}{cc}
\begin{minipage}{0.5\hsize}
\begin{center}
\includegraphics[width=70mm]{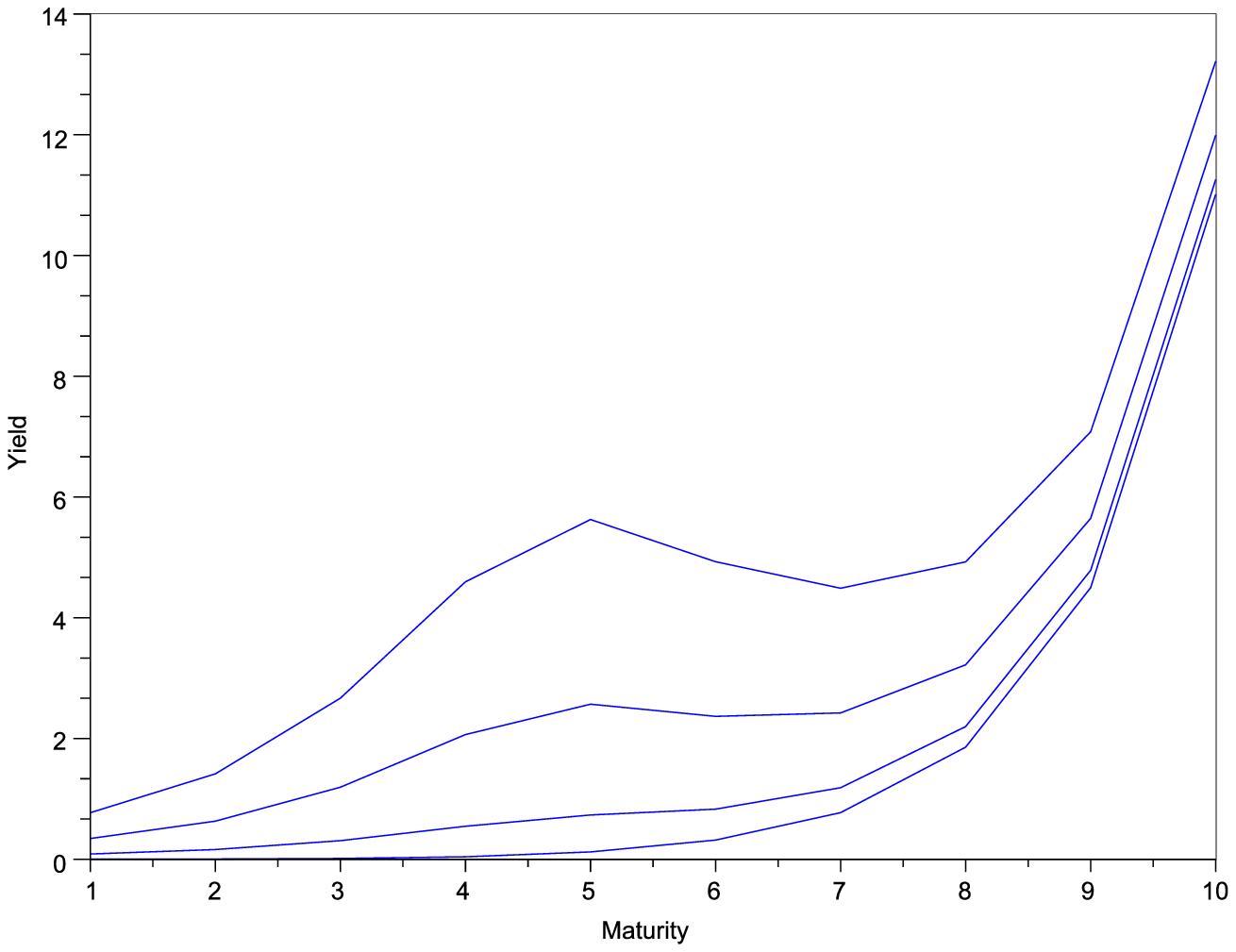}
\caption{Simulated yield curves implied by a default-free bond 
when $\lambda_t = e^{t}/10$}
\label{pfcurve@x}
\end{center}
\end{minipage}
\begin{minipage}{0.5\hsize}
\begin{center}
\includegraphics[width=70mm]{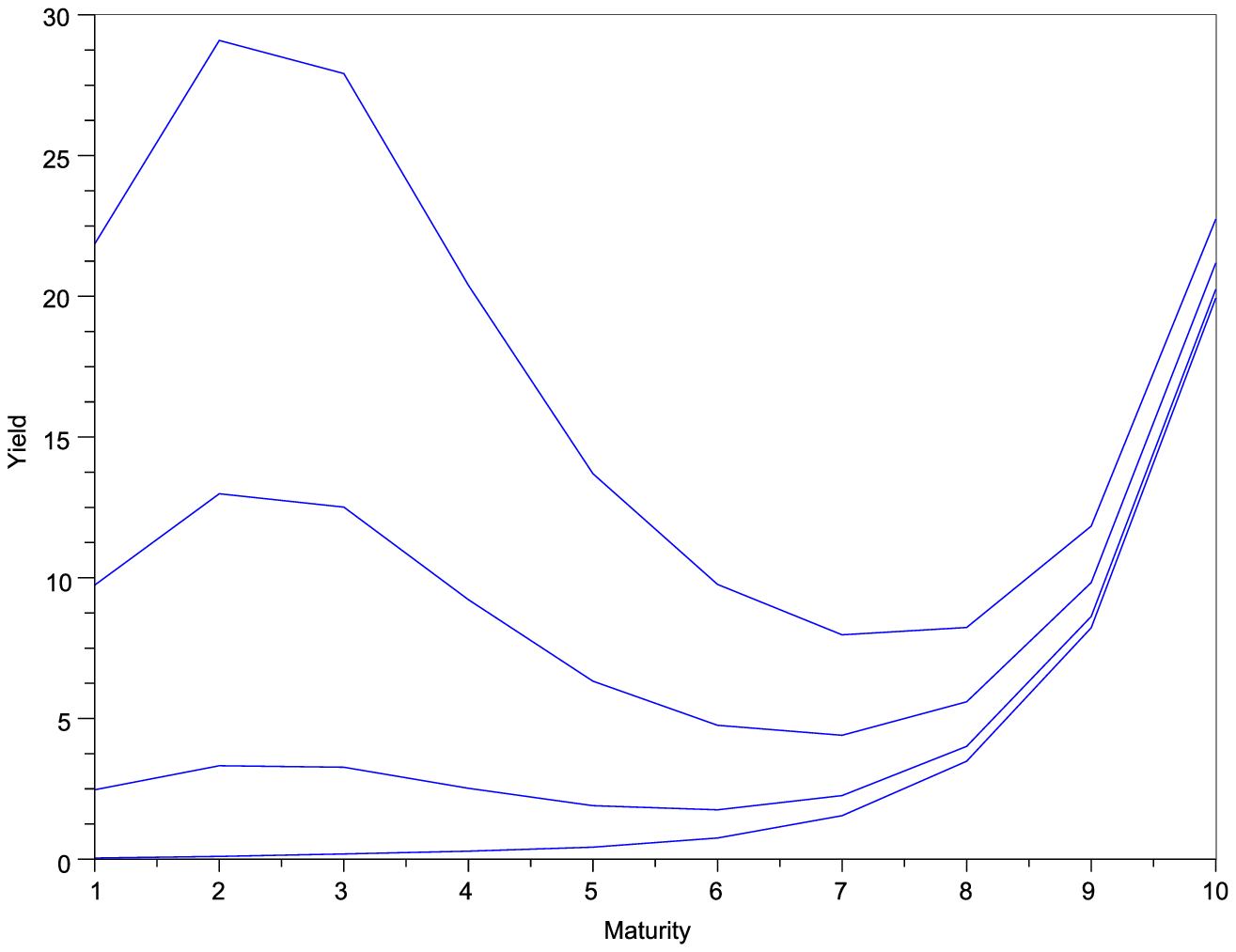}
\caption{Simulated yield curves implied by a default-free bond 
when $\lambda_t = e^{t}/100$}
\label{pfcurve@beta}
\end{center}
\end{minipage}
\end{tabular}
\end{figure}

(Fig.\ref{pfcurve@x}) and (Fig.\ref{pfcurve@beta}) show 
increasing the value of $x$ does not make the curve shift upward, 
but also cause a ``hump" in the curve, 
which can not be observed in the normal affine model, 
with the proper choice of $\lambda$. 
Moreover, using the formula (\ref{crspread}) and (\ref{q^txofqc}), 
we obtain the following simulated credit spread curves as 
(Fig.\ref{figcrsp@beta}). 
Here the parameters are set to be $\lambda_t = \sqrt{t}$, 
$\beta = 0.1, 0.2,...,1$, $x = 0$, and the present time $t = 0$. 

\begin{figure}[htbp]
\begin{center}
\includegraphics[width=70mm]{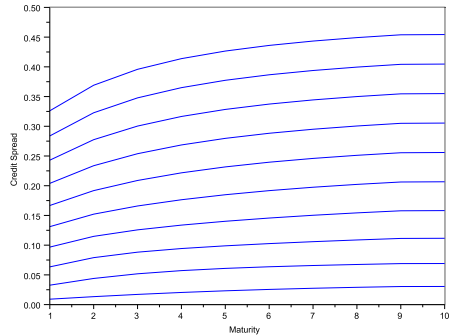}
\caption{Simulated credit spread curves}
\label{figcrsp@beta}
\end{center}
\end{figure}

As usual, the lower the credit rating of a defaultable bond is, 
the wider the spread is, 
and it is non-decreasing in the maturity time.
Moreover, the spread of a defaultable bond lower rated is much wider 
in the maturity time than the one of a bond higher rated.
It should be thought that (Fig.\ref{figcrsp@beta}) shows this fact.

\section{Conclusions}\label{Conc}
We have introduced a way of 
constructing a single defaultable bond market model
under the physical measure $\mathsf{P}$ by applying
the killed HKA. We have also 
presented some simulation results in the quadratic case. 
Comparing the well-known Hull-White model, 
we can have observed a 
complex ``hump" in the yield implied by a default-free bond, 
which comes from the parameter $\lambda$.

\section{Appendix}

\begin{lemma}\label{lemma:d-quadratic}
Let $X$ be a $d$-dimensional Wiener process 
starting at $x$. 
For $ \alpha, \beta \geq 0 $, it holds
\begin{equation}
\begin{split}
& \Ex [e^{-\alpha |X^x_t|^2 -\frac{\beta^2}{2} \int_0^t |X^x_s|^2 \,ds}] \\
&= 
\begin{cases}
\left( \cosh \beta t +\frac{2 \alpha}{\beta} \sinh \beta t \right)^{-d/2}
\exp \left(
-\frac{\beta x^2}{2} 
\frac{{\beta} \sinh \beta t - {2 \alpha} \cosh \beta t}
{ {\beta} \cosh \beta t + {2 \alpha} \sinh \beta t  }
\right)  & \beta > 0, \\
\left( 2 \alpha t + 1  \right)^{-d/2}
\exp \left( -\frac{\alpha x^2}{2 \alpha t + 1} 
\right) & \beta = 0 .
\end{cases}
\end{split}
\end{equation}
\end{lemma}

This is well-known formula and there are many ways to prove it. 
One way is presented in \cite{TI}.

The following is an immediate consequence of Lemma \ref{lemma:d-quadratic}:

\begin{corollary}\label{Corollary:d-quadratic}
Let $X$ be a $d$-dimensional Wiener process 
starting at $x$. 
For $ \beta > 0 $, it holds
\begin{equation*}
\begin{split}
\Ex \left[ e^{-\frac{\beta^2}{2} \int_{s}^{t}  |X^x_v|^2 \,dv } \right]
=&\left( 
\cosh{ \beta( t -s ) } 
+\beta s \sinh{ \beta( t -s ) } 
\right)^{-d/2}\\
&\times \exp \left(
-\frac{ \beta x^2 }{ 2 } \frac{ \sinh{ \beta( t -s )} }
{ \cosh{ \beta( t -s ) } 
+\beta s  \sinh{ \beta( t -s ) } }\, 
\right).
\end{split}
\end{equation*}
\end{corollary}

\begin{proof}
By a Markov property, a Tower property, and Lemma \ref{lemma:d-quadratic},
\begin{equation*}
\begin{split}
\Ex \left[ e^{ -\frac{ \beta^2 }{ 2 } \int_s^t  |X^x_v|^2 \,dv } \right] 
&=\Ex \left[\, 
\Ex \left[ 
e^{ -\frac{ \beta^2 }{ 2 } \int_s^t |X^x_v|^2 \,dv }\, |\, \mathcal{F}^X_s 
\right]\,\, 
\right]
=\Ex \left[\, 
\Ex \left[ 
e^{-\frac{ \beta^2 }{ 2 } \int_0^{t-s} |X^x_v|^2 \,dv }\, |\, X^x_s
\right]\,\, 
\right]\\
&=\cosh{ \beta( t -s ) }^{-d/2}\, 
\Ex \left[ 
\exp \left(
-\frac{ \beta |X^x_s|^2 }{ 2 } 
\frac{ \sinh{ \beta( t -s ) } }{ \cosh{ \beta( t -s )} } 
\right) 
\right].
\end{split}
\end{equation*}
The proof is complete by replacing $\alpha$ by $\frac{ \beta }{ 2 } 
\frac{ \sinh{ \beta( t -s ) } }{ \cosh{ \beta( t -s ) } }$ 
in Lemma \ref{lemma:d-quadratic}.
\end{proof}


\end{document}